\documentclass[envcountsame,final]{svjour3}
\journalname{Quantum Information Processing}

\usepackage[T1]{fontenc}
\usepackage{txfonts}

\begin{document}
\title{Two Gilbert-Varshamov Type Existential Bounds for Asymmetric Quantum Error-Correcting Codes\thanks{This note will be published in Quantum Information
  Processing with DOI:\href{http://dx.doi.org/10.1007/s11128-017-1748-y}{10.1007/s11128-017-1748-y}}}
\titlerunning{Gilbert-Varshamov Bounds for Asymmetric Quantum Codes}
\author{Ryutaroh Matsumoto}
\institute{Ryutaroh Matsumoto \at Department of Communication and Computer Engineering, Nagoya
  University, 464-8603 Japan. \url{http://www.rmatsumoto.org/research.html}}
\date{19 October 2017}
\maketitle
\begin{abstract}
  In this note we report two versions of
  Gilbert-Varshamov type existential bounds for
  asymmetric quantum error-correcting codes.
\end{abstract}

\section{Introduction}
Quantum error-correcting codes (QECC) are important for construction
of quantum computers, as the fault-tolerant quantum computation
is based on QECC \cite{chuangnielsen}.
There are two kinds of errors in quantum information,
one is called a bit error and the other is called a phase error.
Steane \cite{steane96aqc} first studied the asymmetry between
probabilities of the bit and the phase errors, and
he also considered QECC for asymmetric quantum errors,
which are called asymmetric quantum error-correcting codes
(AQECC).
Research on AQECC has become very active recently,
see \cite{galindo16,ioffe06,steane96aqc} and the references therein.

On the other hand,
in the study of error-correcting codes,
it is important to know the optimal performance
of codes.
For classical error-correcting codes,
the Gilbert-Varshamov (GV) bound \cite{macwilliams77}
is a sufficient condition for existence of codes
whose parameters satisfies the GV bound.
By the GV bound, one can know that the optimal
performance of classical codes is at least as good as the
GV bound.

For QECC,
Ekert and Macchiavello obtained
a GV type existential bound for general QECCs.
An important subclass of general QECCs is
the stabilizer codes
\cite{calderbank97,calderbank98,gottesman96},
as they enable efficient encoding and decoding.
Calderbank et al.\ \cite{calderbank97}
obtained a GV type existential bound for
the stabilizer QECCs.
After that,
Feng and Ma \cite{feng04} and
Jin and Xing \cite{jin11}
obtained improved versions of
GV type bounds for the stabilizer QECCs.

The Calderbank-Shor-Steane (CSS) QECCs
\cite{calderbank96,steane96}
are an important subclass of the stabilizer
QECCs, as the CSS codes enable more efficient
implementation of the fault-tolerant quantum
computation than the stabilizer codes.

Those existential bounds
\cite{calderbank97,calderbank96,ekert96,feng04,jin11}
did not consider the asymmetric quantum errors,
while the asymmetry in quantum errors
is important in practice
\cite{sarvepalli09}.
As far as the author know,
nobody has reported
existential bounds for the stabilizer or
the CSS QECC for asymmetric quantum errors.
In this note we report such ones.
Our proof arguments are similar to ones
in \cite{calderbank97,calderbank96}.

\section{A GV type existential bound
  for the CSS codes}
An $[[n,k,d_x,d_z]]_q$ QECC encodes
$k$ $q$-ary qudits into $n$ $q$-ary qudits
and detects up to $d_x$ bit errors and
up to $d_z$ phase errors.
It is known \cite{ashikhmin00,calderbank98}
that a nested classical code 
$C_2 \subset C_1 \subset \mathbf{F}_q^n$
with dimensions $k_2$ and $k_1$ 
can construct an $[[n,\dim C_1-\dim C_2]]_q$
CSS code,
where $\mathbf{F}_q$ is a finite field
with $q$ elements.
A quantum error can be expressed
as a pair $(\vec{e}_x$, $\vec{e}_z)$,
where $\vec{e}_x \in \mathbf{F}_q^n$
corresponds to the bit error component of
a quantum error and
$\vec{e}_x \in \mathbf{F}_q^n$
does to the phase error component.

Let $\mathrm{GL}_n(\mathbf{F}_q)$ be the group of $n\times n$ invertible
matrices over $\mathbf{F}_q$.
Let $B_n = \{ (C_1$, $C_2 ) \mid$
$C_2 \subset C_1 \subset \mathbf{F}_q^n$,
$\dim C_1 = k_1$, $\dim C_2 = k_2\}$.
For a nonzero vector $\vec{e} \in\mathbf{F}_q^n$,
let $B_{n,x}(\vec{e})$ (resp.\ $B_{n,z}(\vec{e})$) be the set of
nested code pairs that cannot detect $\vec{e}$ as
a bit error (resp.\ a phase error), that is,
$B_{n,x}(\vec{e}) = \{ (C_1, C_2) \in B_n \mid
\vec{e} \in C_1 \setminus C_2 \}$
(resp.\ $B_{n,z}(\vec{e}) = \{ (C_1, C_2) \in B_n \mid
\vec{e} \in C_2^\perp \setminus C_1^\perp \}$),
where $C_1^\perp$ is the dual code of $C_1$ with respect to
the standard inner product.

\begin{lemma}
For nonzero $\vec{e}$, we have
\begin{eqnarray*}
\sharp B_{n,x}(\vec{e}) &=& \frac{q^{k_1}-q^{k_2}}{q^n-1}\sharp B_{n},\\
\sharp B_{n,z}(\vec{e}) &=& \frac{q^{n-k_2}-q^{n-k_1}}{q^n-1}\sharp B_{n}.
\end{eqnarray*}
\end{lemma}

\begin{proof}
As each pair $C_2 \subset C_1$ has
$\sharp C_1 \setminus C_2 = q^{k_1}-q^{k_2}$ undetectable errors,
we have
\[
\frac{\sum_{\vec{0}\neq \vec{e}\in \mathbf{F}_q^n} \sharp B_{n,x}(\vec{e})}{\sharp B_{n}}
= q^{k_1}-q^{k_2}.
\]
For nonzero $\vec{e}_1, \vec{e}_2 \in \mathbf{F}_q^n$,
we claim $\sharp B_{n,x}(\vec{e}_1) = \sharp B_{n,x}(\vec{e})_2$.
Assuming the claim, we have
\[
\sum_{\vec{0}\neq \vec{e}\in \mathbf{F}_q^n} \sharp B_{n,x}(\vec{e})
= (q^n-1) \sharp B_{n,x}(\vec{e}).
\]
Combining these two equalities, we have
\[
\sharp B_{n,x}(\vec{e}) = \frac{q^{k_1}-q^{k_2}}{q^n-1}\sharp B_{n}.
\]

We finish the proof by proving the claim.
Let $\vec{e}_1, \vec{e}_2$ be nonzero vectors. We have
\begin{eqnarray*}
  \sharp B_{n,x}(\vec{e}_1) &=& \sharp \{ (C_1, C_2) \in B_n \mid \vec{e}_1 \in C_1 \setminus C_2 \}\\
  &=& \sharp \{ (\tau C_1, \tau C_2)  \mid \tau \in \mathrm{GL}_n(\mathbf{F}_q), \vec{e}_1 \in C_1 \setminus C_2 \}\\
  &=& \sharp \{ (\tau C_1, \tau C_2)  \mid \tau \in \mathrm{GL}_n(\mathbf{F}_q), \tau' \vec{e}_1 \in C_1 \setminus C_2 \}\\
  &=& \sharp \{ ( C_1,  C_2)\in B_n  \mid  \tau' \vec{e}_1 \in C_1 \setminus C_2 \}\\
  &=& \sharp B_{n,x}(\tau'\vec{e}_1),
\end{eqnarray*}
where $\tau' \in \mathrm{GL}_n(\mathbf{F}_q)$ such that $\tau'\vec{e}_1 = \vec{e}_2$.

For phase errors, we can make a similar argument with
$C_2^\perp \supset C_1^\perp$. \qed
\end{proof}

\begin{theorem}\label{th1}
Let $n$, $k_1$, $k_2$, $d_x$ and $d_z$ be positive integers such that
\begin{equation}
\frac{q^{k_1}-q^{k_2}}{q^n-1} \sum_{i=1}^{d_x-1} {n \choose i} (q-1)^i +
\frac{q^{n-k_2}-q^{n-k_1}}{q^n-1}\sum_{i=1}^{d_z-1} {n \choose i} (q-1)^i
< 1, \label{eq1}
\end{equation}
then an
$[[n,k_1-k_2, d_x,d_z]]_q$  CSS QECC exists.
\end{theorem}

\begin{proof}
Recall that each quantum error can be expressed by
its bit error component $\vec{e}_x \in \mathbf{F}_q^n$
and its phase error component $\vec{e}_z \in \mathbf{F}_q^n$.
The bit error component $\vec{e}_x$ cannot be detected by
codes in $B_{n,x}(\vec{e}_x)$ and
the phase error component $\vec{e}_z$ cannot be detected by
codes in $B_{n,z}(\vec{e}_z)$.
The detectabilities of the bit errors and
the phase errors are independent of each other.
Therefore, if Eq.\ (\ref{eq1}) holds
then there exists at least one $(C_1, C_2) \in B_n$ that
can detect all the bit errors with weight up to $d_x-1$ and
all the phase errors with weight up to $d_z-1$, which
implies it is an $[[n,k_1-k_2,  d_x,  d_z]]_q$
quantum code. \qed
\end{proof}

Classical coding theorists often have interest in
asymptotic versions of GV type existential bounds
\cite{macwilliams77}.
They are stated in terms of information rate and
relative distance of classical error-correcting codes.
In the classical error correction, 
information rate is the ratio of the number of information
symbols to the code length, and
relative distance is the ratio of the minimum distance
to the code length.

We can also derive an asymptotic version of
Theorem \ref{th1}.
For an $[[n,k,  d_x,  d_z]]_q$ AQECC,
we may define the relative distance $\delta_x$ for
bit errors as $d_x / n$, and
the relative distance $\delta_z$ for
bit errors as $d_z / n$.
The information rate of an $[[n,k]]_q$ QECC
is defined as $k/n$ \cite{chuangnielsen}.

Recall \cite{macwilliams77} that for $0 \leq \delta \leq 1-1/q$ we have
\begin{equation}
\sum_{i=1}^{\lfloor n\delta \rfloor} {n \choose i} (q-1)^i
\leq q^{nh_q(\delta)}, \label{eq2}
\end{equation}
where $h_q(\delta)=
\delta \log_q(q-1) - \delta \log_q \delta -
(1-\delta)\log_q(1-\delta)$.

\begin{corollary}\label{cor2}
  Let $\delta_x$ and $\delta_z$ be real numbers such that
  $0 \leq \delta_x \leq 1-1/q$ and $0 \leq \delta_z \leq 1-1/q$.
  If
  \begin{eqnarray}
    h_q(\delta_x) &<&1-R_1 ,  \label{eq3}\\
    h_q(\delta_z) & < & R_2, \mbox{ and } \label{eq4}\\
    0 &\leq & R_1-R_2, \nonumber
  \end{eqnarray}
  then, for sufficiently large $n$,
  there exists an $[[n,\lfloor n R_1\rfloor - \lceil nR_2\rceil,\lfloor n\delta_x\rfloor
      ,\lfloor n\delta_z\rfloor]]_q$
  CSS QECC exists.
\end{corollary}
In Corollary \ref{cor2}, $R_1$ is the information rate of
classical ECC $C_1$, and $R_2$ is the information rate of
classical ECC $C_2$. The corresponding
quantum CSS code has information rate $R_1-R_2$,
relative distance $\delta_x$ for bit errors, and
relative distance $\delta_z$ for phase errors.

\begin{proof}
  Assume that Eq.\ (\ref{eq3}) holds.
  Then for sufficiently large $n$ we have
  \begin{eqnarray}
    &&    nh_q(\delta_x) < n-nR_1  \nonumber\\
    &\Rightarrow& q^{nh_q(\delta_x)} < (1/2) \frac{q^n}{q^{nR_1}} \nonumber\\
    &\Rightarrow& \frac{q^{nR_1}}{q^n}q^{nh_q(\delta_x)} < 1/2\nonumber\\
    &\Rightarrow& \frac{q^{\lfloor nR_1\rfloor }-q^{\lceil nR_2\rceil }}{q^n-1}
    \sum_{i=1}^{\lfloor n\delta_x\rfloor -1} {n \choose i} (q-1)^i< 1/2.\label{eq5}
  \end{eqnarray}
  Similarly, for sufficiently large $n$ Eq.\ (\ref{eq4}) implies
  \begin{eqnarray}
    &&    nh_q(\delta_z) < nR_2  \nonumber\\
    &\Rightarrow& q^{nh_q(\delta_z)} < (1/2) \frac{q^n}{q^{n(1-R_2)}} \nonumber\\
    &\Rightarrow& \frac{q^{n(1-R_2)}}{q^n}q^{nh_q(\delta_z)} < 1/2\nonumber\\
    &\Rightarrow& \frac{q^{n- \lceil nR_2\rceil }-q^{n-\lfloor nR_1\rfloor }}{q^n-1}
    \sum_{i=1}^{\lfloor n\delta_z\rfloor -1} {n \choose i} (q-1)^i< 1/2.\label{eq6}
  \end{eqnarray}
  Equations (\ref{eq5}) and (\ref{eq6}) imply that the
  assumption of Theorem \ref{th1} becomes true for sufficiently large $n$,
  which  shows Corollary \ref{cor2}. \qed
\end{proof}

\section{A GV type existential bound
  for the stabilizer codes}
Let $C \subset \mathbf{F}_q^{2n}$ be a $\mathbf{F}_q$-linear space of
dimension $n-k$ self-orthogonal with respect to
the standard symplectic inner product in $\mathbf{F}_q^{2n}$.
$C$ can be viewed as an $[[n,k]]_q$
stabilizer QECC.
Let $A_n$ be the set of all such $C$'s.
A nonzero $\vec{e} \in \mathbf{F}_q^{2n}$
can be viewed as a quantum error on $n$ qudits.
Let $A_n(\vec{e})$ be the set of stabilizer codes
that cannot detect $\vec{e}$ as an error, that is,
$A_n(\vec{e}) = \{ C \in A_n \mid
\vec{e} \in C^{\perp\mathrm{s}} \setminus C \}$,
where $C^{\perp\mathrm{s}}$ is the dual of $C$ with respect to the
symplectic inner product.
Then $\sharp A_n(\vec{e}) \leq
\frac{1-q^{-2k}}{1-q^{-2n}}\cdot \frac{1}{q^{n-k}}\sharp A_n$ \cite[Lemma 9]{matsumotouematsu01}.

Recall that, for $C$ to be $[[n,k,d_x,d_z]]_q$,
$C$ must be able to detect all $d_x$ or less bit errors and
all $d_z$ or less phase errors. The number of such errors is
\[
\sum_{i=1}^{d_x-1} {n \choose i}(q-1)^i \times \sum_{i=1}^{d_z-1} {n \choose i}(q-1)^i.
\]
By the same argument as \cite{matsumotouematsu01}[Remark 10] (or as the last section),
we have the following theorem:
\begin{theorem}\label{th3}
Let $n$, $k_1$, $k_2$, $d_x$ and $d_z$ be positive integers such that
\[
\frac{1-q^{-2k}}{1-q^{-2n}}\cdot \frac{1}{q^{n-k}}
\sum_{i=1}^{d_x-1} {n \choose i} (q-1)^i\times \sum_{i=1}^{d_z-1} {n \choose i} (q-1)^i< 1
\]
then there exists an $[[n,k,d_x,d_z]]_q$
stabilizer QECC. \qed
\end{theorem}

By almost the same argument as Corollary \ref{cor2}
we can derive the following asymptotic version of Theorem \ref{th3}.
\begin{corollary}\label{cor4}
  Let $\delta_x$ and $\delta_z$ be real numbers such that
  $0 \leq \delta_x \leq 1-1/q$ and $0 \leq \delta_z \leq 1-1/q$.
  If
  \begin{equation}
    h_q(\delta_x)+h_q(\delta_z) < 1-R \leq 1 ,\label{eq7}
  \end{equation}
  then, for sufficiently large $n$,
  there exists an $[[n, \lfloor nR\rfloor, \lfloor n\delta_x\rfloor, \lfloor n\delta_z \rfloor]]_q$ stabilizer QECC. \qed
\end{corollary}
The quantum statilizer code in
Corollary \ref{cor4} has information rate $R$,
relative distance $\delta_x$ for bit errors, and
relative distance $\delta_z$ for phase errors.

By the relation between the CSS and the stabilizer QECCs \cite{calderbank98},
we see that the assumption in Corollary \ref{cor2} is less
demanding than that in Corollary \ref{cor4}
for the same $n$, $R=R_1-R_2$, $\delta_x$ and $\delta_z$,
which means that Corollary \ref{cor4} is a stronger
existential bound than Corollary \ref{cor2}.

\begin{remark}
  Theorems \ref{th1} and \ref{th3}, and Corollaries \ref{cor2} and \ref{cor4}
  do not admit direct comparisons against previously known GV type bounds
  even when $d_x = d_z$. The reason is as follows:
  For a binary QECC to be $[[n,k,2,2]]_2$, it must detect at least
  $n^2$ different errors. On the other hand, for a binary $[[n,k]]_2$
  QECC to detect all single symmetric errors, it only has to detect
  $3n$ errors, which is generally much fewer than $n^2$.
  The above example shows that the number of asymmetric quantum errors
  is much different from that of corresponding symmetric quantum errors,
  even if we assume the same number of bit errors and phase errors in
  asymmetric quantum errors.

  In addition, the famous $[[5,1,3]]_2$ binary stabilizer code
  in \cite{calderbank98,gottesman96} can detect up to four bit
  errors if there is no phase error, and can detect up to
  four phase errors if there is no bit error.
  Thus it is simultaneously both $[[5,1,1,5]]_2$ AQECC and
  $[[5,1,5,1]]_2$ AQECC.
  This phenomenon makes the direct comparison even more difficult.
  
\end{remark}

\begin{acknowledgements}
  The author would like to thank an anonymous reviewer
  for careful reading and the helpful report that improved
  this note, and 
  Prof.\ Fernando Hernando for
  drawing his attention to the asymmetric quantum error correction.
\end{acknowledgements}


\begin{thebibliography}{10}
\providecommand{\url}[1]{{#1}}
\providecommand{\urlprefix}{URL }
\expandafter\ifx\csname urlstyle\endcsname\relax
  \providecommand{\doi}[1]{DOI~\discretionary{}{}{}#1}\else
  \providecommand{\doi}{DOI~\discretionary{}{}{}\begingroup
  \urlstyle{rm}\Url}\fi

\bibitem{ashikhmin00}
Ashikhmin, A., Knill, E.: Nonbinary quantum stabilizer codes.
\newblock IEEE Trans. Inform. Theory \textbf{47}(7), 3065--3072 (2001)

\bibitem{calderbank97}
Calderbank, A.R., Rains, E.M., Shor, P.W., Sloane, N.J.A.: Quantum error
  correction and orthogonal geometry.
\newblock Phys. Rev. Lett. \textbf{78}(3), 405--408 (1997)

\bibitem{calderbank98}
Calderbank, A.R., Rains, E.M., Shor, P.W., Sloane, N.J.A.: Quantum error
  correction via codes over {GF(4)}.
\newblock IEEE Trans. Inform. Theory \textbf{44}(4), 1369--1387 (1998)

\bibitem{calderbank96}
Calderbank, A.R., Shor, P.W.: Good quantum error-correcting codes exist.
\newblock Phys. Rev. A \textbf{54}(2), 1098--1105 (1996)

\bibitem{ekert96}
Ekert, A., Macchiavello, C.: Quantum error correction for communication.
\newblock Phys. Rev. Lett. \textbf{77}(12), 2585--2588 (1996)

\bibitem{feng04}
Feng, K., Ma, Z.: A finite {Gilbert-Varshamov} bound for pure stabilizer
  quantum codes.
\newblock IEEE Trans. Inform. Theory \textbf{50}(12), 3323--3325 (2004).
\newblock \doi{10.1109/TIT.2004.838088}

\bibitem{galindo16}
Galindo, C., Geil, O., Hernando, F., Ruano, D.: Improved constructions of
  nested code pairs.
\newblock IEEE Trans. Inform. Theory  (2017).
\newblock \doi{10.1109/TIT.2017.2755682}

\bibitem{gottesman96}
Gottesman, D.: Class of quantum error-correcting codes saturating the quantum
  {H}amming bound.
\newblock Phys. Rev. A \textbf{54}(3), 1862--1868 (1996)

\bibitem{ioffe06}
Ioffe, L., M\'ezard, M.: Asymmetric quantum error-correcting codes.
\newblock Phys. Rev. A \textbf{75}, 032345 (2007).
\newblock \doi{10.1103/PhysRevA.75.032345}

\bibitem{jin11}
Jin, L., Xing, C.: Quantum {Gilbert-Varshamov} bound through symplectic
  self-orthogonal codes.
\newblock In: Proc. 2011 IEEE International Symposium on Information Theory,
  pp. 455--458. Sait Petersburg, Russia (2011).
\newblock \doi{10.1109/ISIT.2011.6034167}

\bibitem{macwilliams77}
MacWilliams, F.J., Sloane, N.J.A.: The Theory of Error-Correcting Codes.
\newblock Elsevier, Amsterdam (1977)

\bibitem{matsumotouematsu01}
Matsumoto, R., Uyematsu, T.: Lower bound for the quantum capacity of a discrete
  memoryless quantum channel.
\newblock Journal of Mathematical Physics \textbf{43}(9), 4391--4403 (2002).
\newblock \doi{10.1063/1.1497999}

\bibitem{chuangnielsen}
Nielsen, M.A., Chuang, I.L.: Quantum Computation and Quantum Information.
\newblock Cambridge University Press, Cambridge, UK (2000)

\bibitem{sarvepalli09}
Sarvepalli, P.K., Klappenecker, A., R\"otteler, M.: Asymmetric quantum codes:
  constructions, bounds and performance.
\newblock Proceedings of the Royal Society A: Mathematical, Physical and
  Engineering Sciences \textbf{465}(2105), 1645--1672 (2009).
\newblock \doi{10.1098/rspa.2008.0439}

\bibitem{steane96}
Steane, A.M.: Multiple particle interference and quantum error correction.
\newblock Proc. Roy. Soc. London Ser. A \textbf{452}(1954), 2551--2577
  (1996)

\bibitem{steane96aqc}
Steane, A.M.: Simple quantum error-correcting codes.
\newblock Phys. Rev. A \textbf{54}(6), 4741--4751 (1996).
\newblock \doi{10.1103/PhysRevA.54.4741}

\end{thebibliography}

\end{document}